\documentclass{daj}

\usepackage{amsmath,amsfonts,amssymb,amsthm,mathrsfs,mathtools}
\usepackage{enumitem}

\newcommand{\e}{\varepsilon}
\newcommand{\vphi}{\varphi}
\newcommand{\iy}{\infty}
\newcommand{\st}{\ : \ }

\renewcommand{\leq}{\leqslant}

\renewcommand{\geq}{\geqslant}

\DeclareMathOperator{\tr}{Tr}

\DeclareMathOperator{\Id}{I}
\DeclareMathOperator{\I}{Id}

\DeclareMathOperator{\conv}{conv}

\DeclareMathOperator{\E}{\mathbf{E}}

\DeclareMathOperator{\card}{card}

\newcommand{\cD}{\mathrm{D}}

\newcommand{\cF}{\mathcal{F}}
\newcommand{\cH}{\mathcal{H}}

\newcommand{\cN}{\mathcal{N}}

\newcommand{\cS}{\mathrm{Sep}}

\newcommand{\PSD}{\mathcal{PSD}}

\newcommand{\N}{\mathbb{N}}

\newcommand{\R}{\mathbb{R}}
\newcommand{\C}{\mathbb{C}}

\renewcommand{\P}{\mathbf{P}}

\newcommand{\cM}{\mathsf{M}}
\newcommand{\M}{\cM}

\newcommand{\gGL}{\mathsf{GL}}

\newcommand{\mC}{\mathcal{C}}

\newcommand{\scalar}[2]{\langle #1 , #2\rangle}
\newcommand{\braket}[2]{\langle #1 | #2\rangle}
\newcommand{\ketbra}[2]{| #1 \rangle \langle #2 |}
\newcommand{\bra}[1]{\langle #1 |}
\newcommand{\ket}[1]{| #1 \rangle}
\newcommand{\sa}{\textnormal{sa}}

\newcommand{\HS}{\textnormal{HS}}
\newcommand{\op}{\textnormal{op}}

\theoremstyle{plain}
\newtheorem{theorem}{Theorem}

\newtheorem{proposition}[theorem]{Proposition}
\newtheorem{lemma}[theorem]{Lemma}

\theoremstyle{definition}

\theoremstyle{remark}

\dajAUTHORdetails{
  title = {Dvoretzky's Theorem and the Complexity of Entanglement Detection},
  author = {Guillaume Aubrun and Stanis\l aw Szarek},
  plaintextauthor = {Guillaume Aubrun and Stanislaw Szarek},
  keywords = {Complexity of entanglement, Figiel--Lindenstrauss--Milman inequality, Dvoretzky's theorem, Facial dimension, Verticial dimension},
}   

\dajEDITORdetails{%
   year={2017},
   number={1},
   received={5 April 2016},   
   revised={9 December 2016},    
   published={12 January 2017},  
   doi={10.19086/da.1242},       
}   

\begin{document}

\begin{frontmatter}[classification=text]

\author{Guillaume Aubrun\thanks{Supported in part by ANR (France) 
grants OSQPI (2011-BS01-008-02) and StoQ (2014-CE25-0003)}}
\author{Stanis\l aw Szarek\thanks{Supported in part by grants from the National Science
Foundation (U.S.A.) and by the first ANR grant listed above}}

\begin{abstract}
The well-known Horodecki criterion asserts that a state $\rho$ on $\C^d \otimes \C^d$ is entangled if and only if
there exists a positive map $\Phi : \cM_d \to \cM_d$ such that the operator $(\Phi \otimes \I)(\rho)$ is not positive
semi-definite. 
We show that the number of such maps needed to detect all robustly entangled states 
(i.e., states $\rho$ which remain entangled even in the presence of substantial randomizing noise) 
exceeds $\exp(c d^3 / \log d)$. 
The proof is based on the 1977 inequality of Figiel--Lindenstrauss--Milman,   
which ultimately relies on Dvoretzky's theorem 
about almost spherical sections of convex bodies.      
We interpret that inequality as a statement about approximability of 
convex bodies  by polytopes with few vertices or with few faces and 
combine it with the study of fine properties of the set of quantum states and that of separable states. 
 Our results can be thought of as 
geometrical manifestations of the complexity of entanglement detection.
\end{abstract}
\end{frontmatter}

\section{Introduction}
Entanglement \cite{EPR35, Schrodinger35, Werner89} lies at the heart of quantum mechanics 
and is a fundamental resource for quantum information and computation. 
It underlies many of the most striking potential 
applications of quantum phenomena to information processing such as, for example, teleportation \cite{Bennett93}. 
However, its properties remain elusive; even at the mathematical level, the current understanding 
of entanglement in high-dimensional systems remains very incomplete, and not for lack of trying. 
In particular, there is an extensive literature on the entanglement detection, of which we 
mention below just several highlights. 

It is an elementary observation that if $\rho$ is a separable state on $\C^d \otimes \C^d$ and 
$\Phi : \cM_d \to \cM_d$ is a positive map (i.e., a map which preserves positive semi-definiteness 
of $d \times d$ matrices), then 
$(\Phi \otimes \I)(\rho)$ is positive semi-definite. 
A remarkable result known as the Horodecki criterion \cite{HHH96} asserts that the converse is true: 
if a state $\rho$ on $\C^d \otimes \C^d$ is entangled, then there exists a positive map 
$\Phi : \cM_d \to \cM_d$ which detects its entanglement 
in the sense that $(\Phi \otimes \I)(\rho)$ has a negative eigenvalue. 
Such a map is called an \emph{entanglement witness}. 

The study of positive maps between matrix algebras is notoriously difficult. 
The situation is quite simple when $d=2$: 
any positive map on $\cM_2$ is decomposable \cite{Stormer63}, i.e.,  can be written
as $\Phi_1 + \Phi_2 \circ T$ where $\Phi_1,\Phi_2$ are completely positive maps and 
$T$ is the transposition on $\cM_2$. 
(Of course completely positive maps by themselves are useless for the task of entanglement 
detection since all their extensions are positive by definition.)  
It follows that the well-known Peres partial transposition criterion  
is a necessary and sufficient condition for separability of $2$-qubits states \cite{Peres96,HHH96}. 

The situation in higher dimensions is much less clear. 
To describe it, we will use the following concept.
Let $\cF = (\Phi_i)$ be a family of positive maps on $\cM_d$ and let $\mathrm{E}$ 
be a subset of the set of entangled states on $\C^d \otimes \C^d$.
We say that $\cF$ is \emph{universal} for $\mathrm{E}$ if for every $\rho \in \mathrm{E}$, 
there is an index $i$ such that $\Phi_i$ is an entanglement witness for $\rho$, i.e., 
$(\Phi_i \otimes \I)(\rho)$ has a negative eigenvalue. 

First, for $d \geq 3$, the 
partial transposition criterion is no longer sufficient \cite{PH97}. Moreover, for such $d$, 
any family $\cF$ which is universal for all entangled states must be infinite 
(this result is proved in \cite{Skowronek16} and based on \cite{HaKye11}). 

However, asking for detecting \emph{all} entangled states is probably too demanding for any practical purpose. 
We say that a state $\rho$ on $\C^d \otimes \C^d$
is \emph{robustly entangled} if $\frac{1}{2}(\rho + \rho_*)$ is entangled, 
where $\rho_* = \Id/d^2$ denotes the maximally mixed state. In other words, 
robustly entangled states remain entangled even in the presence of substantial randomizing noise. 
The main result of the paper is a super-exponential lower bound on the cardinality of 
any universal family which detects robust entanglement.

\begin{theorem} \label{theorem:main} 
There is a universal constant $c>0$ such that the following holds. 
Consider $d \geq 2$ and let $(\Phi_i)_{1 \leq i \leq N}$ be a family of positive maps on $\cM_d$ 
which is universal for all robustly entangled states. 
Then $N+1 \geq \exp( c d^3 / \log d)$.
\end{theorem}

We used the factor $\frac{1}{2}$ in the definition of robust entanglement only for simplicity; 
the same proof works for any fixed choice of weights.
With a little care, the argument gives actually much more. 

\medskip

\noindent {\bf Theorem 1'.} \emph{There are universal constants $c_0$, $c>0$ such that the following holds. Consider $d \geq 2$ and set $\e_d=c_0 \log d / \sqrt{d}$. 
Let $(\Phi_i)_{1 \leq i \leq N}$ be a family of positive maps on $\cM_d$ which is universal for the set 
\begin{equation} \label{eq:extremely-entangled}
\{ \rho \textnormal{ state on } \C^d \otimes \C^d \st \e_d \rho + (1-\e_d) \rho_* \textnormal{ is entangled} \}. 
\end{equation}
Then $N+1 \geq \exp( c d^2 \log d)$.}

\medskip

The fact that universal families must be large is not surprising. 
Indeed, each positive map leads to a test for entanglement detection which runs in polynomial time. 
Consequently, the existence of small universal families would have to be reconciled with the 
known result that deciding whether a given state is separable or entangled is an NP-hard problem.  
This was first observed by Gurvits \cite{Gurvits03} and refined in \cite{Ioannou07,Gharibian10}; 
other relevant references include \cite{BCY11, HarrowMontanaro13, GHMW15}. 
However, to the best of our knowledge, results in the spirit of Theorem \ref{theorem:main} 
cannot be derived from the existing literature. For starters, the complexity results cited  
above that address lower bounds generally focus on states situated 
\emph{very} close to the separability/entanglement border  
(which precludes the robustness feature present in our setting) 
or are based on computational assumptions such as in  \cite{HarrowMontanaro13}. 
An exception is the forthcoming work \cite{HNW16}, 
which addresses lower bounds on the size of some relaxations of entanglement detection problems via semidefinite programming.
(This includes in particular the hierarchy of extendible states from \cite{DPS04}; see \cite{Lancien16} for related questions.)
Let us also mention a result from \cite{SWZ08} that is similar in spirit to ours: 
the set of completely positive maps occupies a subexponentially (in the dimension of that set) 
small proportion, in terms of volume,  of the set of all positive maps.

Our proof of Theorem \ref{theorem:main} is geometric and is based on the following observation due to 
Figiel--Lindenstrauss--Milman \cite{FLM77}: an $n$-dimensional polytope 
with a center of symmetry cannot have few faces and -- simultaneously -- 
 few vertices. Complexity must lie somewhere. 
 This paradigm is actually rather general and can be applied to the set
of separable states (although it is neither a polytope nor centrally symmetric): 
given that it has ``few'' extreme points, it must have many ``faces.''  
Since we can upper-bound the number of faces provided by each test detecting entanglement, we 
conclude that many tests are needed. 
This vague scheme can be converted to a rigorous proof through the introduction of the 
\emph{verticial dimension} and the \emph{facial dimension} of a convex body $K$, 
which quantify the number of vertices
(resp., faces) required for a polytope to approximate $K$ within a constant factor, and 
are measures of algorithmic complexity of $K$.  
In the process, we obtain sharp bounds for some of these invariants for the set   
of all quantum states and for the set of separable states; these bounds 
(particularly \eqref{eq:dimV-D} and \eqref{eq:dimV-S})  and  
 the arguments leading to them are  
surprisingly subtle and may be of independent interest.  

\medskip
The paper is organized as follows. 
The remainder of the Introduction is devoted to the notation and to basic background results. 
Section \ref{section:dim} introduces the concepts of verticial and facial
dimensions, and states the fundamental Figiel--Lindenstrauss--Milman 
inequality \eqref{eq:FLM} asserting that their product must be large. 
It contains, in Table \ref{table}, estimates of 
these parameters for a selection of classical convex bodies, 
and for the set of all quantum states and that of separable states. 
The latter estimates constitute the main technical ingredient of the proofs of 
Theorems \ref{theorem:main} and \ref{theorem:main}' 
that are presented in in Section \ref{section:proof}. 
The Figiel--Lindenstrauss--Milman inequality is related to the classical 
Dvoretzky--Milman theorem in Section \ref{section:dvoretzky}. 
Section \ref{section:dim-DS} proves the estimates stated in Table \ref{table} 
(sometimes up to a logarithmic factor, when irrelevant for the main argument). 
Finally, in Section \ref{section:dimension-D-random} 
we complete the proof of the full strength of the 
bounds on verticial dimensions from Table \ref{table} 
(stated as Theorems \ref{theorem:dimension-D} and 
\ref{theorem:dimension-S} in Section \ref{section:dim-DS}, 
and specifically as \eqref{eq:dimV-D} and \eqref{eq:dimV-S}) 
by  removing the remaining ``technical'' logarithmic factor. 

The results from this paper will be incorporated in a forthcoming book \cite{book},  
which contains more background
on both Quantum Information Theory and Asymptotic Geometric Analysis, 
and many examples of their interaction.

\subsubsection*{Notation and basic facts} 
A convex body $K \subset \R^n$ is an $n$-dimensional convex compact set.
Denote by $|\cdot|$ the Euclidean norm in $\R^n$ or $\C^n$, and by $B_2^n$ 
and $S^{n-1}$ the unit ball and unit sphere in $\R^n$.
An $\e$-net in a set $S \subset \R^n$ 
is a subset $\cN \subset S$ with the property that for any $x \in S$ 
there is $y \in \cN$ with $|x-y| \leq \e$. We will repeatedly use the following elementary bound.
\begin{lemma} \label{lemma:net}
For every $\e \in (0,1)$ and $n \in \N$, there is an $\e$-net $\cN$ in $S^{n-1}$ with $\card \cN \leq (1+2/\e)^n$. 
Conversely, if $\cN$ is an $\e$-net in $S^{n-1}$ for some $\e \in (0,\sqrt{2})$, then $\card \cN \geq 2/\sin^{n-1} \theta$,
where $\theta = 2 \arcsin(\e/2) < \pi/2$ is the angle between two points in $S^{n-1}$ which are $\e$-distant.
\end{lemma}

The first part of Lemma \ref{lemma:net} is proved by a volumetric argument (see \cite{Pisier89}, Lemma 4.10). The second part follows from the fact
that the proportion of $S^{n-1}$ covered by a spherical cap of angular radius $\theta$ is less than $\frac{1}{2}(\sin \theta)^{n-1}$, as
can be checked by simple geometric considerations.

The unit sphere in $\C^m$ is denoted by $S_{\C^m}$. Since $S_{\C^m}$ identifies with $S^{2m-1}$ as a metric space, the results from Lemma
\ref{lemma:net} also apply. 
We denote by $\cM_m$ the algebra of complex $m \times m$ matrices, which we identify with operators on $\C^m$.
We use Dirac bra-ket notation. In particular, 
 if $\psi \in S_{\C^m}$, then $\ketbra{\psi}{\psi}$ denotes the rank 1 orthogonal projection onto $\C \psi$. 
 The inner product of
vectors $\psi,\varphi$ is denoted $\braket{\psi}{\varphi}$, and if $A \in \cM_m$ we write $\bra{\psi}A\ket{\varphi}$ for $\braket{\psi}{A(\varphi)}$. Such notation leads to
visually pleasant formulas such as $\tr ( \ketbra{\psi}{\psi} A ) = \bra{\psi}A \ket{\psi}$.
A fundamental object, which we denote by $\cD(\C^m)$ 
or simply by $\cD$ when the context is clear, is the set of (mixed) states on $\C^m$ defined as 
\begin{eqnarray*} 
\cD (\C^m) & \coloneqq & \{ \rho \in \cM_m \st \rho = \rho^\dagger,\ \rho \geq 0,\ \tr \rho = 1 \} 
= \conv \{ \ketbra{\psi}{\psi} \st \psi \in S_{\C^m} \}.
\end{eqnarray*}
States of the form $\ketbra{\psi}{\psi}$ are called pure states and coincide with the set of extreme points of $\cD$. 
We call maximally mixed state the state $\rho_*\coloneqq\Id/m$, where $\Id$ is the identity matrix.
When $\C^m$ is identified
with the tensor product $\C^d \otimes \C^d$ (with $m=d^2$), we denote by $\cS(\C^d \otimes \C^d)$, or simply $\cS$, 
the subset of $\cD$ formed by separable states:
\[ \cS \coloneqq \conv \left\{ \ketbra{\psi \otimes \vphi}{\psi \otimes \vphi} \st \psi,\vphi \in S_{\C^d} \right\}. \]
States which are not separable are called entangled.
Both $\cD$ and $\cS$ live in the affine space
\begin{equation} \label{eq:def-H}
H = \{ A \in \cM_m  \st A=A^\dagger, \ \tr A = 1 \}. 
\end{equation}
In order to use tools from geometry of normed spaces (a.k.a.\ \emph{asymptotic geometric analysis}), we consider $H$ as a vector space whose origin 
is the maximally mixed state $\rho_*=\Id/m$. 
We use $\bullet$ to denote scalar multiplication in $H$ when thought of 
as a vector space, i.e.,
for $\rho \in H$ and $t \in \R$,
\begin{equation} \label{eq:bullet} 
t \bullet \rho \coloneqq t \rho + (1-t) \rho_* .
\end{equation}
If $K \subset H$, then denote $t \bullet K = \{ t \bullet x \st x \in K \}$.
We will repeatedly use the following fact: for convex bodies $K$, $L$ in $H$ and $t \geq 0$, the inclusion 
$t \bullet K \subset L$ is equivalent to the inequality
\begin{equation}
\label{eq:tKsubsetL}
 t \sup_{\rho \in K} \tr (A \rho) \leq \sup_{\rho \in L} \tr (A \rho)
\end{equation}
holding for every trace zero Hermitian matrix $A$.

We equip $H$ with the Hilbert--Schmidt (a.k.a. Frobenius) norm 
 $\|A\|_{\HS}=(\tr A^2)^{1/2}$ inherited from 
$\cM_m$, so that the unit ball is
$B_{\HS} \coloneqq \{ A \in H \st \|A-\rho_*\|_{\HS} \leq 1 \}$. 
Denote also by $\|\cdot\|_{\op}$ the usual operator (or spectral) norm and by 
$\|A\|_{\textnormal{Tr}} = \tr ((AA^\dagger)^{1/2} )$ the trace-class norm. 

\section{Verticial and facial dimension of convex sets} \label{section:dim}

Let $K \subset \R^n$ be a convex body  containing $0$ in the interior. Fix a number $A>1$, our resolution parameter. 
All polytopes we consider are convex.
Define the \emph{verticial dimension} of $K$ as
\[ \dim_V(K,A) \coloneqq \log
\inf \{ N : \textnormal{ there is a polytope }P\textnormal{ with }N\textnormal{ vertices s.t.\ } K \subset P \subset AK \}
\]
and the \emph{facial dimension} of $K$ as
\[ \dim_F(K,A) \coloneqq \log
\inf \{ N : \textnormal{ there is a polytope }Q\textnormal{ with }N\textnormal{ facets s.t.\ } K \subset Q \subset AK \},
\]
where by facets we mean faces of dimension $n-1$. We set the resolution parameter $A$ as a default value equal
to $4$ and write 
\[\dim_V(K)\coloneqq\dim_V(K,4) \ \textnormal{  and  } \ \dim_F(K)\coloneqq\dim_F(K,4).\] 
All the results below are only affected in the values of the numerical constants (implicit in the notation $O(\cdot)$, $\Theta(\cdot)$ and $\Omega(\cdot)$) 
if $4$ is replaced by another number larger that $1$. 
However, the character of the dependence on $A$ will be important in some applications.  

We note that the above concepts are linear invariants in the following sense: $\dim_V(TK) = \dim_V(K)$ and 
$\dim_F(TK) = \dim_F(K)$ for any $T \in \gGL(n,\R)$. Moreover, there are dual to each other: if we define the
\emph{polar} of a convex body 
$K \subset \R^n$ (say, containing $0$ in the interior) as the convex body
\[ K^\circ = \{ x \in \R^n \st \scalar{x}{y} \leq 1 \textnormal{ for all } y \in K \} ,\]
then 
\[
\dim_V(K^\circ) = \dim_F(K)  \textnormal{  and }  \dim_F(K^\circ) = \dim_V(K) .
\]
 Indeed, $P$ is a polytope with $N$ facets
if and only if $P^\circ$ is a polytope with $N$ vertices.

We also note that if $E \subset \R^n$ is a linear subspace, 
$\dim_F(K \cap E) \leq \dim_F K$ and $\dim_V(P_E K) \leq \dim_V K$ where
$P_E$ denotes the orthogonal projection onto $E$. 
These inequalities reflect the fact that projections do not increase the number of vertices
of polytopes, while sections do not increase the number of facets.

For any convex body $K \subset \R^n$ which is $0$-symmetric 
(i.e., such that $K=-K$), we have $\dim_V(K) = O(n)$ and
$\dim_F(K)=O(n)$ by a standard volumetric argument (see, e.g., Lemma 1 in \cite{AubrunLancien15}). 
This fails in complete generality without the symmetry assumption, but for wrong reasons: 
consider the case of a disk in $\R^2$ which contains
the origin \emph{very}  close to its boundary. 
If we insist that, for example, $K$ has centroid at the origin, then the inequalities 
$\dim_V(K) = O(n)$ and $\dim_F(K) = O(n)$ still hold, but this is
less obvious than in the symmetric case 
(see \cite{Bronstein08,Barvinok14,LRT14,Szarek14,book} for this and related questions).  

Define also the \emph{asphericity} of a convex body $K \subset \R^n$ as 
\begin{equation}
\label{eq:def-a} 
 a(K) = \inf \left\{ \frac{R}{r} \st \textnormal{there is a 0-symmetric ellipsoid } \mathcal{E} \textnormal{ with }
r \mathcal{E} \subset K \subset R \mathcal{E} 
 \right\}. 
 \end{equation} 
 
The reader will notice that, arguably, it would be more natural and more functorially sound to define 
$\dim_F(\cdot)$, $\dim_V(\cdot)$ 
and $a(\cdot)$ with an additional infimum over all translates of $K$: we would end up then with 
\emph{affine} invariants (and not just \emph{linear} invariants).  
However, this is not necessary in our setting and would in fact lead to complications in duality considerations.   
 
We will use in a fundamental way the following result, which appears in \cite{FLM77} only implicitly (see the paragraphs preceding
Example 3.5). We also 
make explicit the dependence on resolution parameters.

\begin{theorem} \label{theorem:FLM}
For any convex body $K \subset \R^n$ containing the origin in the interior we have 
\begin{equation} \label{eq:FLM-4} 
\dim_F(K) \dim_V(K) \, a(K)^2 =\Omega(n^2) .
\end{equation}
More generally, if $A,B > 1$, then 
\begin{equation} \label{eq:FLM} 
A^2  \dim_F(K,A) \cdot B^2 \dim_V(K,B) \cdot \, a(K)^2 =\Omega(n^2) .
\end{equation}
\end{theorem} 

Theorem \ref{theorem:FLM} is fairly sharp for many  convex bodies, as can be seen from Table \ref{table} below. 
Moreover, for all those examples it is enough to consider in \eqref{eq:def-a}  
the appropriate Euclidean balls rather than general ellipsoids $\mathcal{E}$. 
In a nutshell, our argument to prove Theorems \ref{theorem:main} and \ref{theorem:main}'
-- which is presented in the next section -- combines the upper bound on $\dim_F(\cD)$ 
with the lower bound on $\dim_F(\cS)$,
 the latter being obtained as a consequence of Theorem \ref{theorem:FLM}.
The proofs of these bounds are relegated to Sections \ref{section:dim-DS} and \ref{section:dimension-D-random}.

\begin{table}[htbp]
\caption{\label{table} Parameters appearing in \eqref{eq:FLM-4} for some families of convex bodies, 
see Section \ref{section:dim-DS} for justifications and references. The main technical points of this  
paper are the estimates from the last two rows on the verticial and facial dimensions of the set of states (Theorem \ref{theorem:dimension-D}) and of the
set of separable states (Theorem \ref{theorem:dimension-S}). 
The bounds implicit in the first three rows are either trivial or well-known and are included here mostly 
to provide reference points. }
\begin{center}
\begin{tabular}{|c|c|c|c|c|}
  \hline
  $K$ & dimension & $a(K)$ & $\dim_V(K)$ & $\dim_F(K)$ \\
  \hline
  $\vphantom{\displaystyle \sum} \textnormal{Euclidean ball } B_2^n$ & $n$ & $1$ & $\Theta(n)$ & $\Theta(n)$ \\
  \hline
  $\vphantom{\displaystyle \sum} \textnormal{Cube } [-1,1]^n$ & $n$ & $\sqrt{n}$ & $\Theta(n)$   & $\Theta(\log n)$ \\
  \hline
  $\vphantom{\displaystyle \sum} \textnormal{Simplex in }\R^n$ & $n$ & $n$ & $\Theta(\log n)$ & $\Theta(\log n)$ \\
  \hline
  $\vphantom{\displaystyle \sum} \textnormal{Set of quantum states } \cD(\C^m)$ & $m^2-1$ & $m-1$ & $\Theta(m)$ & $\Theta(m)$ \\
  \hline
  $\vphantom{\displaystyle \sum} \textnormal{Set of separable states } \cS(\C^d \otimes \C^d)$ & $d^4-1$ & $d^2-1$ & $\Theta(d \log d)$ & $\Omega(d^3/\log d)$ \\
  \hline
  \end{tabular}
\end{center}
\end{table}

As noted earlier, the values of the invariants appearing in Table \ref{table} depend crucially 
on the location of the origin, which is not canonical for non-centrally-symmetric bodies. 
For the simplex, we assume the origin to coincide with the centroid. 
In particular, if we  think of the $n$-dimensional simplex
as the set of classical states (i.e., probability measures) on \ $n+1$ points, 
the role of the origin is played by the uniform probability measure with weights 
$\big(\frac{1}{n+1},\cdots,\frac{1}{n+1}\big)$. 
This is analogous to the quantum case, 
where the maximally mixed state $\rho_*$ 
is considered as the origin; the choice being implicit in definition \eqref{eq:bullet} of the operation $\bullet$.  

For future reference, we point out that the set
$\cD=\cD(\C^m)$ of states on $\C^m$ satisfies the relation
\begin{equation} \label{eq:D-polar} \cD^\circ = (-m) \bullet \cD \end{equation}
(polarity in $H$,  with $\rho_*$ as the origin).  This is a consequence
of the  self-duality of the cone of positive semi-definite matrices. 
It follows in particular that $\dim_F(\cD)=\dim_V(\cD)$.

\section{Proof of the main theorems} \label{section:proof}

In this section we prove Theorems \ref{theorem:main} and \ref{theorem:main}' 
as consequences of the estimates appearing in Table \ref{table}. 
In particular, we rely on a lower bound on the facial dimension of the set of separable states 
(Theorem \ref{theorem:dimension-S})
that will be proved in Section \ref{section:dim-DS}.

In what follows, the symbol $\cD$ will always stand for the set of states on 
$\cH=\C^d \otimes \C^d$ and  $\cS \subset \cD$ for
the corresponding set of separable states. Next, $B^{\sa} = B^{\sa}(\cH)$ 
will denote the space of self-adjoint operators on $\cH$, while 
$\PSD \subset B^{\sa}$ will be the cone of positive semidefinite operators on $\cH$. 
Let $\{\Phi_1,\dots,\Phi_N\}$ be a family of $N$ positive maps on $\cM_d$ 
which satisfies the hypothesis of Theorem \ref{theorem:main} or Theorem \ref{theorem:main}'.
This is equivalent to the right hand side  inclusion in 
\begin{equation} \label{eq:witness-intersection2}
\cS \  \subset  \ \ \bigcap_{i=1}^N 
 \left\{ \rho \in \cD \st (\Phi_i \otimes \I)(\rho) \in \PSD \right\} \ \subset \  A \bullet \cS.
\end{equation}
where either $A=2$ (Theorem \ref{theorem:main}) or 
$A=1/\e_d=c_0^{-1} \sqrt{d}/\log d$ (Theorem \ref{theorem:main}'). 
[The  left hand side  inclusion is the easy part of the Horodecki criterion.] 
The value of the (universal positive) constant $c_0$ will be determined at the end of the proof. 

The idea of the argument is to show that each of the sets appearing under 
the intersection in \eqref{eq:witness-intersection2} can be well-approximated by a polytope 
with ``not too many'' facets.   
Since the number of facets of a polytope is subadditive under intersections, 
this leads to an upper bound on the facial dimension of $\cS$. 
Finally, we can compare this upper bound with the corresponding estimate from Table \ref{table}, 
which will lead to a lower bound on $N$. 

To that end, we note first that we can assume that $\Phi_i(\Id)$ is invertible for every $i$. 
Indeed, if this is not the case, denote by $E \subsetneq \C^d$ 
the range of $\Phi_i(\Id)$ (which is a positive operator) and replace $\Phi_i$ by  
$\tilde{\Phi}_i : X \mapsto \Phi_i(X) + P_{E^\perp} X P_{E^\perp}$. 
The map $\tilde{\Phi}_i$ is clearly positive and has the property that, for any state 
$\rho$  on $\C^d \otimes \C^d$, we have 
\[ (\Phi_i \otimes \I)(\rho) \in \PSD \iff (\tilde{\Phi}_i \otimes \I)(\rho) \in \PSD .\] 
(The key point in inferring the latter is that positivity of $\Phi_i$ implies then that, 
for any $X \in \cM_d$, the range of $\Phi_i(X)$ is contained in $E$.)  Further, we can 
also assume that $\Phi_i$ is unital (i.e., that $\Phi_i(\Id) = \Id$) 
by considering the map $X \mapsto \Phi_i(\Id)^{-1/2} \Phi_i(X) \Phi_i(\Id)^{-1/2}$. 

Next, we prove a simple lemma about approximability of $\cD$ by polytopes. It does not imply the result
stated in Table \ref{table} but is sufficient for our present purposes.

\begin{lemma} \label{lemma:approximation-D}
Let $\mathcal{M}$ be a $\delta$-net in $(S_{\C^m},|\cdot|)$. Then 
\begin{equation} \label{eq:approximation-D-bad}
(1-2 m \delta) \bullet  \cD(\C^m) \subset \conv \{ \ketbra{\psi_i}{\psi_i} \st \psi_i \in \mathcal{M} \} \subset \cD(\C^m).
\end{equation}
\end{lemma}

The reader will notice that the proof given below can be fine-tuned to yield a slightly better 
-- but more complicated -- factor $\big(1-2 (m-1) \delta\big)$ in \eqref{eq:approximation-D-bad}.
\begin{proof}
By \eqref{eq:tKsubsetL}, we have to show that, for any trace zero Hermitian matrix $A$, 
\[ \lambda_1(A) \coloneqq\sup_{\psi \in S_{\C^m}} \bra{\psi} A \ket{\psi} \leq (1-2\delta m)^{-1} \sup_{\psi_i \in \mathcal{M}} 
\bra{\psi_i} A \ket{\psi_i}. \]
Since $A$ has trace $0$, we have $\|A\|_{\op} \leq m \lambda_1(A)$. 
Given $\psi \in S_{\C^m}$, there is $\psi_i \in \mathcal{M}$ with $|\psi-\psi_i| \leq \delta$. 
By the triangle inequality, we have
\begin{eqnarray} \label{eq:netD1} \bra{\psi} A \ket{\psi} & \leq & \delta \|A\|_{\op} + \bra{\psi} A \ket{\psi_i} \\
\label{eq:netD2} & \leq & 2\delta \|A\|_{\op} + \bra{\psi_i} A \ket{\psi_i} \\
& \leq  &2\delta m\lambda_1(A)  + \bra{\psi_i} A \ket{\psi_i}. 
\end{eqnarray}
Taking supremum over $\psi$, we get 
$\lambda_1(A)  \leq 2\delta m\,\lambda_1(A)  + \sup \{ \bra{\psi_i} A \ket{\psi_i} \st \psi_i \in \mathcal{M} \}$ 
and the result follows.
\end{proof}

We now set $\e=1/(1+d)$ and 
$\delta = \e/2d^2$, and choose  a $\delta$-net $\cN$ in $S_{\C^d \otimes \C^d}$. By
Lemma \ref{lemma:net}, we may assume that $\log \card \cN = O( d^2 \log d)$.
We know from Lemma \ref{lemma:approximation-D}
that $(1-\e) \bullet \cD \subset Q \subset \cD$, where $Q$ is the polytope
\[ Q = \conv \{ \ketbra{\psi}{\psi} \st \psi \in \cN \} .\]
It follows from \eqref{eq:D-polar} that the polytope $P \coloneqq (-(1-\e)d^{-2}) \bullet Q^\circ$, 
which has at most $\card \cN$ facets, satisfies 
\begin{equation} \label{eq:P-vs-D} (1-\e) \bullet \cD \subset P \subset \cD. \end{equation} 

It is now instructive to complete the argument under the additional assumption that each 
$\Phi_i$ is also trace-preserving. 
Since $\Phi_i \otimes \I$ is then likewise trace-preserving, 
the condition $(\Phi_i \otimes \I)(\rho) \in \PSD$  from \eqref{eq:witness-intersection2} 
is equivalent to $\rho \in  (\Phi_i \otimes \I)^{-1}(D)$ and so, in view of \eqref{eq:P-vs-D},
\[
(\Phi_i \otimes \I)^{-1}(P) \ \subset \    \left\{ \rho \st (\Phi_i \otimes \I)(\rho) \in \PSD \right\}  \ \subset \  
(1-\e)^{-1} (\Phi_i \otimes \I)^{-1}(P).
\]
This shows that we succeeded in approximating the sets from \eqref{eq:witness-intersection2} 
 by polyhedra with $\exp\big(O( d^2 \log d)\big)$ facets, as required for the heuristics 
we sketched earlier.  Note that the additional constraint $\rho \in D$ can be handled 
in a formal way by adding to the family $\{\Phi_i\}$ the map $\Phi_0=\I$, 
and that $\Phi_i$ being unital translates to $(\Phi_i\otimes \I)(\rho_*)=\rho_*$, which  
assures that we are $\bullet$-dilating all sets with respect to the same point. 

\smallskip The general case requires some tweaking:  we need to be able to control how far 
$\Phi_i$ and $\Phi_i\otimes \I$ are from being trace-preserving. 
We will use the following 

\begin{lemma} \label{lemma:unital-not-too-wild}
Let $\Phi : \cM_d \to \cM_d$ be a positive unital map. Then, for any $\rho \in \cD$,
\[ 0 \leq \tr \left[ (\Phi \otimes \I) \rho \right] \leq d. \] 
\end{lemma}

\begin{proof}
Since linear forms achieve their extrema on extreme points of convex compact sets, we may assume that 
$\rho= \ketbra{\psi}{\psi}$ is pure. Let $\psi = \sum \lambda_i e_i \otimes f_i$ the Schmidt decomposition of
$\psi$.  Then, by direct calculation,  
\[ \tr \left[ (\Phi \otimes \I) \rho \right] = \sum_{i=1}^d \lambda_i^2 \tr \Phi( \ketbra{e_i}{e_i}) \leq d , \]
 the last inequality following from 
$\sum \lambda_i^2 =1$ and from $\Phi ( \ketbra{e_i}{e_i} ) 
\leq \Phi(\Id) = \Id$.
\end{proof}

Returning to the proof of the theorems, we denote the convex bodies appearing in \eqref{eq:witness-intersection2} by  
\begin{equation} \label{eq:def-K_i}
K_i \coloneqq  \left\{ \rho \in \cD \st (\Phi_i \otimes \I)(\rho) \in \PSD \right\} = \cD \cap (\Phi_i \otimes \I)^{-1}(\PSD) 
\end{equation}
(note that $\rho_* \in K_i$) and define the polyhedral cones
\begin{equation} \label{eq:def-C_i}
 \mC_i \coloneqq \left\{ M \in B^{\sa}
\st (\Phi_i \otimes \I)(M) \in \R_+ P \right\}. 
\end{equation}
We now claim that
\begin{equation} \label{eq:approx-K_i}
\frac 12 \bullet K_i \subset P \cap \mC_i \subset K_i. 
\end{equation} 

Before proving the claim, let us first show how it implies the Theorems. 
Combining \eqref{eq:approx-K_i} and \eqref{eq:witness-intersection2} we obtain
\[
\frac 12 \bullet  \cS \  \subset  \ \bigcap_{i=1}^N \left(\frac 12 \bullet  K_i \right)  \  \subset  \  \bigcap_{i=1}^N \left(P \cap \mC_i \right)  \  =  \ P\, \cap \, \bigcap_{i=1}^N \mC_i 
\ \subset \ \bigcap_{i=1}^N K_i  \ \subset \  A \bullet \cS .
 \]
 The polytope $R = P \cap \bigcap_{1 \leq i \leq N} \mC_i $ 
 has at most $f\coloneqq(N+1) \exp(C d^2 \log d)$ facets 
 (i.e., $N+1$ times the number of facets of $P$), 
 where $C$ is the constant implicit in the notation $\log \card \cN = O( d^2 \log d)$. 
Consequently, by the definition of the facial dimension, 
 we must have $\log f \geq \dim_F(\cS,2A)$ and so 
\[ \log(N+1) + C d^2 \log d = \log f \geq \dim_F(\cS,2A) \geq c d^3 A^{-2} / \log d ,  \]
where the last inequality will be proved as Theorem \ref{theorem:dimension-S} in Section \ref{section:dim-DS} (for $A=2$, this is exactly the statement
from Table \ref{table}). 
In the case of Theorem \ref{theorem:main} ($A=2$), the conclusion is immediate. 
In the case of Theorem \ref{theorem:main}'
($A=c_0^{-1} \sqrt{d}/\log d$), we choose $c_0 = \sqrt{2C/c}$ and conclude that
$ \log (N+1) \geq C d^2 \log d$, as asserted.

It remains to prove the claim \eqref{eq:approx-K_i}. 
The second inclusion is immediate from the definitions and from \eqref{eq:P-vs-D}. 
For the first inclusion, 
it is clearly enough to show that $\frac 12 \bullet K_i \subset  \mC_i$. 
To that end, let $\rho \in K_i$ and denote 
$t = \tr \left[ (\Phi_i \otimes \I) \rho \right] \geq 0$. 
We now consider two cases. 
First, if  $t=0$, then -- since  $(\Phi_i \otimes \I)(\rho)$ is a positive operator -- 
we must have $(\Phi \otimes \I)(\rho)=0$.  Hence 
 trivially $\rho  \in \mC_i$ and, \emph{a fortiori}, 
 $\frac{1}{2} \bullet \rho \in  \mC_i$.  
If $t>0$,  we note that $t^{-1} (\Phi_i \otimes \I)(\rho) \in \cD$ and that, 
by Lemma \ref{lemma:unital-not-too-wild}, we have $ t \leq d$,
and therefore $\frac{t}{1+t} = 1-\frac{1}{1+t} \leq 1-\frac{1}{1+d} = 1-  \e$. 
It thus follows from \eqref{eq:P-vs-D} that
\[ \frac{t}{1+t}  \bullet t^{-1} (\Phi_i \otimes \I)(\rho )\ \in \  \frac{t}{1+t}  \bullet D \ \subset \ (1-\e) \bullet D \ \subset \ P .\]
It remains to notice that 
\[\frac{t}{1+t}  \bullet t^{-1} (\Phi_i \otimes \I)(\rho ) = 
\frac{(\Phi_i \otimes \I)(\rho ) +\rho_*}{1+t} = \frac{2}{1+t}   ( \Phi_i \otimes \I )\left( \frac{\rho + \rho_*}{2} \right) ,\]
which means that we showed that $( \Phi_i \otimes \I )\left( \frac{1}{2} \bullet \rho \right) \in \frac{1+t}{2} \, P$. 
In particular (cf.\ \eqref{eq:def-C_i}), $\frac{1}{2} \bullet \rho \in \mC_i$, as needed.  

\section{Connection with Dvoretzky's theorem} \label{section:dvoretzky}

Theorem \ref{theorem:FLM} is actually a consequence of Milman's version of Dvoretzky's theorem \cite{Dvoretzky61,Milman71},  which gives a sharp formula for
the dimension of almost spherical sections of convex bodies, together with the fact that the Euclidean ball has large facial and verticial dimensions, as shown by
the following lemma. We make explicit the dependence on the resolution parameter.

\begin{lemma} \label{lemma:ball} 
For $n \geq 1$ and $A>1$, we have $\dim_V(B_2^n,A) =\dim_F(B_2^n,A)\geq \frac{n-1}{2A^2}$. 
\end{lemma} 
\begin{proof} First, since $(B_2^n)^\circ=B_2^n$, the verticial and the facial dimensions coincide; 
this justifies the equality in the assertion. 
Now assume  that $\dim_V(B_2^n,A)= \log N$, so that there exists 
a polytope $P$ with vertices $x_1,\dots,x_N$ such that $A^{-1} B_2^n \subset P \subset B_2^n$. 
There is no loss of generality in assuming that $x_i \in S^{n-1}$. A simple geometric argument shows then that the spherical caps centered at
$x_i$ and of angle $\theta = \arccos (A^{-1})$ cover $S^{n-1}$. It follows from Lemma \ref{lemma:net} that $N \geq 2/\sin^{n-1} \theta$, and it remains
to use to the elementary inequality $\sin (\arccos x) \leq \exp(-x^2/2)$ valid for $-1 \leq x \leq 1$.
\end{proof}

Recall that to each convex body $K \subset \R^n$ containing
the origin in the interior, we may associate its gauge defined for $x \in \R^n$ as 
\[ \|x\|_K \coloneqq \inf \{ t \geq 0 \st x \in tK \}. \]
This gauge is a norm if and only if $K$ is $0$-symmetric.

\begin{theorem} [Dvoretzky--Milman theorem] 
\label{theorem:dvoretzky-tangible}
There is an absolute constant $c>0$ such that the following holds. 
Let $K$ be a convex body in $\R^n$ such that 
$rB_2^n \subset K$ for some $r>0$, and let $\e \in (0,1)$. 
Denote by $M$ the expectation of $\|X\|_K$, where $X$ is a uniformly distributed random vector on $S^{n-1}$. 
Then there exist an integer  $k \geq c \e^2 M^2 r^2 n$ and a $k$-dimensional
subspace $E \subset \R^n$ such that
\[ (1-\e) M B_2^E \subset K \cap E \subset (1+\e) M B_2^E \]
where $B_2^E \coloneqq B_2^n \cap E$ denotes the unit ball in $E$.
\end{theorem}

Theorem \ref{theorem:dvoretzky-tangible} is a fundamental result in the geometry 
of high-dimensional convex bodies. If
we do not insist on having the correct dependence on $\e$ (which was shown in \cite{Gordon88,Schechtman89}, 
but which is not needed here), its proof essentially amounts to using
concentration of measure in the form of L\'evy's lemma \cite{Levy51}, 
combined with a simple union bound argument. We also note 
that the hypothesis that $K$ is symmetric present in \cite{Milman71} is not needed in the argument. 
Another important point is that
the conclusion of Theorem \ref{theorem:dvoretzky-tangible} holds for most subspaces $E$, 
but this aspect is not relevant to the present paper.
An application of (the complex version of) Theorem \ref{theorem:dvoretzky-tangible} 
in Quantum Information Theory appears in \cite{ASW11},  where 
it is shown to imply and conceptually simplify Hastings's result \cite{Hastings09} 
about non-additivity of classical capacity of quantum channels.

For the reader's convenience, and because the statement is only implicit in \cite{FLM77}, 
we reproduce the argument 
allowing to derive Theorem \ref{theorem:FLM} from Theorem \ref{theorem:dvoretzky-tangible}. 
Let $K$ be a convex body
in $\R^n$ containing $0$ in the interior. 
Since the verticial and facial dimension are invariant under linear transformations, we may
assume that the ellipsoid witnessing the infimum in \eqref{eq:def-a} is a Euclidean ball, i.e., that 
$rB_2^n \subset K \subset RB_2^n$ with $R/r=a(K)$.  Let $M = \E \|X\|_K$ 
and $M^* = \E \|X\|_{K^\circ}$ where 
$X$ is a random vector uniformly distributed on the unit sphere. 
The pointwise inequality $\|\cdot\|^{1/2}_K \|\cdot\|^{1/2}_{K^\circ} \geq 1$ 
implies by the Cauchy--Schwarz inequality that $MM^* \geq 1$.

We now apply Theorem \ref{theorem:dvoretzky-tangible} to $K$ with $\e=1/2$ (say). 
This yields a subspace $E \subset \R^n$ of dimension $\Omega((rM)^2n)$ such that 
\[ \frac{M}{2} B_2^E \subset K \cap E \subset \frac{3M}{2} B_2^E .\]
It follows that $\dim_F(K \cap E,A) \geq \dim_F(B_2^E,3A) \geq \frac{1}{18A^2} (\dim E-1)$, 
where the second inequality comes from Lemma \ref{lemma:ball}.
Consequently 
\[\dim_F(K,A) \geq \dim_F(K \cap E,A) = \Omega((rM)^2n A^{-2}). \]
We apply the same argument to $K^\circ$
(note that $R^{-1} B_2^n \subset K^\circ$) and obtain that 
\[\dim_F(K^\circ,B) = \Omega((M^*/R)^2nB^{-2}). \] 
Since $\dim_V(K,B)=\dim_F(K^\circ,B)$, it follows that
\[ \dim_F(K,B) \dim_V(K,A) = \Omega \left(n^2 (MM^*)^2 (r/R)^2 A^{-2}B^{-2}\right) 
= \Omega \left( n^2 A^{-2}B^{-2}/ a(K)^2 \right), \]
as needed. 

\section{Facial and verticial dimension of \texorpdfstring{$\cD$}{D} and \texorpdfstring{$\cS$}{Sep}} 
\label{section:dim-DS}

We now give references or justifications for the values appearing in Table \ref{table}. 
The case of the Euclidean ball is essentially contained in Lemma \ref{lemma:net} 
(cf. Lemma \ref{lemma:ball}).
The estimates for the cube and for the simplex are  either trivial or follow by standard 
arguments; they are  not used in this paper. 

We set $\cD = \cD(\C^m)$ and $\cS = \cS(\C^d \otimes \C^d)$. We recall that the role of the origin
is played by the maximally mixed state, i.e., that the condition on polytopes appearing in the definition of
$\dim_V(\cD)$ or $\dim_F(\cD)$ is $P \subset \cD \subset 4 \bullet P$, and similarly for $\cS$.
 It is elementary to check that 
\begin{equation} \label{eq:HS-D} 
\frac{1}{\sqrt{m(m-1)}} \bullet B_{\HS} \subset \cD \subset \sqrt{\frac{m-1}{m}} \bullet B_{\HS} 
\end{equation}
so that $a(\cD) \leq m-1$. We have actually $a(\cD) = m-1$: by a symmetry argument, the optimal ellipsoid must
be a multiple of the Hilbert--Schmidt ball, and the values in \eqref{eq:HS-D} are optimal.  

Similarly, we have
\begin{equation} \label{eq:HS-S} 
\frac{1}{\sqrt{d^2(d^2-1)}} \bullet B_{\HS} \subset \cS \subset \sqrt{\frac{d^2-1}{d^2}} \bullet B_{\HS} . 
\end{equation}
While the second inclusion in \eqref{eq:HS-S} is an immediate consequence of \eqref{eq:HS-D}  and 
of $\cS$ being a subset of $\cD$, the first one is a non-trivial result due to 
Gurvits and Barnum \cite{GurvitsBarnum02}. It follows that $a(\cS) \leq d^2-1$ and, like for $\cD$,  
there is actually an equality. 

In order to justify all the estimates appearing in Table \ref{table}, we prove Theorems \ref{theorem:dimension-D} and \ref{theorem:dimension-S} below.
The lower bound on the facial dimension of $\cS$ is obtained in an indirect way via the Figiel--Lindenstrauss--Milman inequlity.
This is reminiscent of the arguments from 
\cite{ASY12,ASY14}, where calculating \emph{directly} certain invariants of the set  $\cS$  was not feasible 
because of the hardness of detecting entanglement, but it was possible to reasonably estimate 
the values of those invariants using duality considerations and deep results from asymptotic geometric analysis. 

\begin{theorem} \label{theorem:dimension-D}
Let $\cD = \cD(\C^m)$. We have 
\begin{equation} \label{eq:dimV-D} 
d_F(\cD)=d_V(\cD) = \Theta (m). 
\end{equation}
\end{theorem}

\begin{theorem} \label{theorem:dimension-S}
Let $\cS = \cS(\C^d \otimes \C^d)$. We have 
\begin{equation} \label{eq:dimV-S} 
\dim_V (\cS) = \Theta(d \log d) \end{equation}
and
\begin{equation} \label{eq:dimF-S} 
\dim_F(\cS) = \Omega(d^3/\log d). \end{equation}
More generally, for any $A>1$, 
\begin{equation} \label{eq:dimF-S-A} \dim_F(\cS,A) = \Omega(d^3A^{-2}/\log d). \end{equation}
\end{theorem}

As we already noted, the self-duality of $\cD$ (see \eqref{eq:D-polar}) implies that
$\dim_V(\cD) = \dim_F(\cD)$. Since $a(\cD)=m-1$, Theorem \ref{theorem:FLM} implies that 
\[ d_F(\cD)=d_V(\cD) = \Omega (m) .\]
It is also easy to supply a direct argument going along the same lines as (but simpler than) 
the proof of the lower bound in \eqref{eq:dimV-S} presented later in this section. 

Surprisingly, the upper bound $d_V(\cD) = O(m)$ is not that easy to establish and 
so we postpone the proof of Theorem \ref{theorem:dimension-D} to 
Section \ref{section:dimension-D-random}.  However, the slightly weaker bound 
$O(m \log m)$, which is sufficient for the proof of Theorems 
Theorem \ref{theorem:main} and \ref{theorem:main}',
 follows immediately from Lemma \ref{lemma:approximation-D} (applied with $\delta = \frac{3}{8m}$)
and Lemma \ref{lemma:net}.

We now prove Theorem \ref{theorem:dimension-S}. We first note that since $a(\cS)=d^2-1$, 
once we know that \eqref{eq:dimV-S} holds, \eqref{eq:dimF-S} and \eqref{eq:dimF-S-A} 
follow by an application of Theorem \ref{theorem:FLM}.
It is likely that the lower bound \eqref{eq:dimF-S} on $\dim_F(\cS)$ is not sharp; 
any improvement would reflect on 
the estimate for $N$ in Theorem \ref{theorem:main}.

\begin{proof}[Proof of the upper bound in \eqref{eq:dimV-S}]
Let $\cN$ be an $\e$-net in $(S_{\C^d},|\cdot|)$ with $\e$ to be determined. We want to show that
\begin{equation} \label{eq:inclusion-S} \frac{1}{4} \bullet \cS  \subset \conv 
\left\{ \ketbra{\psi_i \otimes \psi_j}{\psi_i \otimes \psi_j} 
\st \psi_i,\psi_j \in \cN \right\} \subset \cS.\end{equation} 
Equivalently, by \eqref{eq:tKsubsetL}, we must show that for any trace zero Hermitian matrix $A$ we have
\[ W\coloneqq\sup_{\psi,\vphi \in S_{\C^d}} \bra{\psi \otimes \vphi} A \ket{\psi \otimes \vphi} 
\leq 4 \sup_{\psi_i,\psi_j \in \cN} \bra{\psi_i \otimes \psi_j} A 
\ket{\psi_i \otimes \psi_j}. \]
First, note that using the left inclusion from \eqref{eq:HS-S} yields 
\[ W \geq \frac{1}{d^2} \|A\|_{\HS} \geq \frac{1}{d^2} \|A\|_{\op}. \] 
Given $\vphi,\psi \in S_{\C^d}$, there are $\psi_i,\psi_j \in \cN$ with 
$|\vphi-\psi_i| \leq \e$
and $|\psi-\psi_j| \leq \e$. Using the triangle inequality as in \eqref{eq:netD1}--\eqref{eq:netD2}, we have
\[ \bra{\vphi \otimes \psi} A \ket{\vphi \otimes \psi} 
\leq 4\e \|A\|_{\op} + \bra{\psi_i \otimes \psi_j} A \ket{\psi_i \otimes \psi_j} 
\leq 4\e d^2 W + \bra{\psi_i \otimes \psi_j} A \ket{\psi_i \otimes \psi_j}. \]
Taking supremum over $\vphi,\psi$, we obtain
\[ W \leq 4\e d^2 W + \sup_{\psi_i,\psi_j \in \cN} \bra{\psi_i \otimes \psi_j} A 
\ket{\psi_i \otimes \psi_j}. \]
We now set $\e={3}/{16d^2}$; this guarantees that \eqref{eq:inclusion-S} holds. 
By Lemma \ref{lemma:net}, we may choose $\cN$ such that $\card \cN \leq (16d^2)^{2d}$. 
Since we produced a polytope $P$ with $(\card \cN)^2$ vertices 
such that $\frac{1}{4} \bullet P \subset \cS \subset P$, it follows that $\dim_V(\cS) = O(d \log d)$. 
\end{proof} 

The estimates used in the argument above may appear quite crude and so it comes as a surprise that 
the obtained bound is actually tight.

\begin{proof}[Proof of the lower bound in \eqref{eq:dimV-S}]
Let $P$ be a polytope with $N$ vertices such that $\frac{1}{4} \bullet \cS \subset P \subset \cS$. 
By Carath\'eodory's theorem, we may write
each vertex of $P$ as a combination of $d^4$ extreme points of $\cS$ (which are pure product states, i.e., of
the form $\ketbra{\psi \otimes \vphi}{\psi \otimes \vphi}$ for unit vectors  $\psi, \vphi \in \C^d$). 
We obtain therefore a polytope
$Q$ which is the convex hull of $N' \leq Nd^4$ pure product states, 
and such that $\frac{1}{4} \bullet \cS \subset P \subset Q \subset \cS$.
Let $(\ketbra{\psi_i \otimes \vphi_i}{\psi_i \otimes \vphi_i})_{1 \leq i \leq N'}$ be the vertices of $Q$.
Fix $\chi \in S_{\C^d}$ arbitrarily. For any $\vphi \in S_{\C^d}$, 
let $\alpha = \max \{ |\braket{\vphi}{\vphi_i}|^2 \st 1 \leq i \leq N' \}$. Consider the 
linear form
\[ g (\rho) = \tr \left[ \rho \left( \ketbra{\chi}{\chi} \otimes (\alpha \Id_{\C^d} - \ketbra{\vphi}{\vphi})\right) \right] . \]
For any $1 \leq i \leq N'$ we have
\[ g ( \ketbra{\psi_i \otimes \vphi_i}{\psi_i \otimes \vphi_i}) = 
|\braket{\chi}{\psi_i}|^2(\alpha-|\braket{\vphi}{\vphi_i}|^2) \geq 0 \]
and therefore $g$ is nonnegative on $Q$. Since $Q \supset \frac{1}{4} \bullet \cS$, we have
\begin{eqnarray*} 0 \leq g\left( \frac{1}{4} \bullet \ketbra{\chi \otimes \vphi}{\chi \otimes \vphi}  \right) & = 
&  \frac{1}{4}\, g(\ketbra{\chi \otimes \vphi}{\chi \otimes \vphi}) + \frac{3}{4} \,g(\rho_*) \\
&=&\frac{1}{4}( \alpha -1) + \frac{3}{4} \times \frac{1}{d} \left(\alpha-\frac{1}{d} \right) \\
&=& \frac{1}{4} \alpha \left( 1 + \frac{3}{d}\right) - \frac{1}{4} \left( 1 + \frac{3}{d^2} \right).
\end{eqnarray*}
It follows that
\[\alpha \geq \frac{1+\frac{3}{d^2}}{1+\frac{3}{d}} \geq 1-\frac{3}{d}. \]
In other words, we showed that for every $\vphi \in S_{\C^d}$ there is an index $i \in \{1,\dots,N'\}$ such that
$| \braket{\vphi}{\vphi_i} |^2 \geq {1-3/d}$.
This means that $(\vphi_i)_{1 \leq i \leq N'}$ is an $O(1/\sqrt{d})$-net in the projective space
over $\C^d$, when equipped with the quotient metric from $(S_{\C^d},|\cdot|)$. It follows that the set
\[ \cN = \left\{ e^{2 \imath \pi j/d} \vphi_i \st 1 \leq i \leq N', 1 \leq j \leq d \right\} \]
is an $O(1/\sqrt{d})$-net in $S_{\C^d}$. Thus, by Lemma \ref{lemma:net},  
$\card \cN \geq (c\sqrt{d})^{2d-1}$ for some 
absolute constant $c>0$. At the same time, $\card \cN \leq dN' \leq d^5 N$, 
 and combining the two bounds we infer that $\log N = \Omega ( d \log d)$, as asserted. 
\end{proof}

\section{The verticial dimension of \texorpdfstring{$\cD$}{D}: final touches} \label{section:dimension-D-random}

In this section we will complete the proof of Theorem \ref{theorem:dimension-D}. 
A proof of the lower bound (actually one proof and a sketch of another proof) was given in 
Section \ref{section:dim-DS}, after the statement of Theorem \ref{theorem:dimension-S}. 
Concerning the upper bound, it 
may seem reasonable to expect that choosing $\cN$ as a $\delta$-net in $S_{\C^d}$ (for some
sufficiently small $\delta$ independent of $d$) and taking the convex hull of the corresponding states 
would yield a polytope 
$P$ such that $\frac{1}{4} \bullet \cD \subset P \subset \cD$. 
This idea works for the unit ball for the trace class norm -- the
``symmetrized'' version of $\cD$ -- see Lemma 3 in \cite{AubrunSzarek06}.  
What makes the problem intriguing is that this approach fails for $\cD$. 
Indeed, given $\delta$, for $d$ large enough, a $\delta$-net $\cN$ may have
the property that for some fixed unit vector $\psi$, we have  $|\braket{\vphi_i}{\psi}|  > 1/\sqrt{d}$
for every $\vphi_i \in \cN$. It follows that, for every $\rho \in \conv \{ \ketbra{\vphi_i}{\vphi_i} \}$, 
we have $\bra{\psi} \rho \ket{\psi}> 1/d$. 
However, this inequality fails for $\rho = \rho_*$, which shows that even the maximally mixed
state does not belong to the convex hull of the net! 
Elements of the net may somehow conspire towards the direction $\psi$. 

Yet, this approach can be salvaged if we use a balanced $\delta$-net to avoid such conspiracies.
Lemma \ref{lemma:approximation-D} is not enough to directly imply Theorem \ref{theorem:dimension-D}, 
but it can be bootstrapped to yield the needed estimate. 
The idea is to use -- instead of an arbitrary net -- a family 
of random points independently and uniformly distributed on the unit sphere, 
and to show that these points satisfy
the conclusion of Theorem \ref{theorem:dimension-D} with high probability. 
(The observation that randomly chosen subsets often form very efficient nets goes back at least to 
Rogers \cite{Rogers57,Rogers63}.)
We actually prove the following, which
gives Theorem \ref{theorem:dimension-D} by specializing to $\e=3/4$.

\begin{proposition} \label{proposition:approximation-D}
For every $\e \in (0,1)$, there is a constant $C(\e)$ such that the following holds: for every dimension $d \geq 2$, there exists
a family $\cN = (\vphi_i)_{1 \leq i \leq N}$ of unit vectors in $\C^d$, with $N \leq \exp(C(\e) d)$, such that
\begin{equation} \label{eq:approximation-D}
 (1-\e) \bullet \cD(\C^d) \subset \conv \{ \ketbra{\vphi_i}{\vphi_i} \st \vphi_i \in \cN \}.
 \end{equation}
\end{proposition}

\begin{proof}[Proof of Proposition \ref{proposition:approximation-D}]
The conclusion \eqref{eq:approximation-D} can be (by \eqref{eq:tKsubsetL}) 
equivalently reformulated as follows: 
\emph{for any self-adjoint trace zero matrix $A$ we have}  
\begin{equation} \label{eq:bound-lambdamax} 
\lambda_{1}(A) = \sup_{\psi \in S_{\C^d}} \bra{\psi} A \ket{\psi} 
\leq \frac{1}{1-\e} \; 
\sup_{\vphi_i \in \cN} \bra{\vphi_i} A \ket{\vphi_i} .
\end{equation}
Let $\mathcal{M}$ be an $\frac{\e}{4d}$-net in $S_{\C^d}$ given by Lemma \ref{lemma:net}.   
By Lemma \ref{lemma:approximation-D}, we have
\begin{equation} \label{equation:approx1}
\sup_{\psi \in S_{\C^d}} \bra{\psi} A \ket{\psi} \leq \frac{1}{1-\e/2} \; \sup_{\psi \in \mathcal{M}} \bra{\psi} A \ket{\psi} .
\end{equation}
Set $\theta = \sqrt{\e/8}$. For $\psi \in S_{\C^d}$, denote by $C(\psi,\theta) \subset S_{\C^d}$ 
the cap with center $\psi$ and radius $\theta$ with respect to the geodesic distance. 
By symmetry, there is a number $\alpha$ 
(depending on $d$ and $\e$) such that
\begin{equation} 
 \label{eq:average-cap} 
\frac{1}{\sigma(C(\psi,\theta))} \int_{C(\psi,\theta)} \ketbra{\vphi}{\vphi} \, \mathrm{d}  \sigma(\vphi) 
= (1-\alpha) \bullet \ketbra{\psi}{\psi}
\end{equation}
where $\sigma$ denotes the uniform probability measure on $S_{\C^d}$.
Taking (Hilbert-Schmidt) inner product with $\ketbra{\psi}{\psi}$, we obtain
\[ 1-\alpha +\frac{\alpha}{d} = \frac{1}{\sigma(C(\psi,\theta))} \int_{C(\psi,\theta)}| \braket{\psi}{\vphi} |^2 
\, \mathrm{d}  \sigma(\vphi)  \geq \cos^2 \theta \geq 1 - \theta^2 \]
so that 
\begin{equation}  \label{eq:average-cap-weight} 
\alpha \leq \theta^2\frac{d}{d-1} \leq \e/4. 
\end{equation}
Denote $L \coloneqq \sigma(C(\psi,\theta))^{-1}$ and let $\mathcal{N} = \{\vphi_i \st 1 \leq i \leq N \}$ 
be a family of $N=\lceil 2L^3\rceil$ independent random vectors
uniformly distributed on $S_{\C^d}$. 
(To not to obscure the argument, we will pretend in what follows 
that $2L^3$ is an integer and so $N=2L^3$.) 
We will rely on the following lemma (to be proved later). 

\begin{lemma} \label{lemma:approximation-D-by-random}
Let $B_{\op} = \{ \Delta \in \cM_d \st \|\Delta\|_{\op} \leq 1 \}$ be the unit ball for the operator norm. 
For $\psi \in S_{\C^d}$ and $t \geq 0$, the event
\[ 
E_{\psi,t} = \{( \vphi_i )  \st (1-\alpha) \bullet \ketbra{\psi}{\psi}  \in t B_{\op} 
+ \conv \{ \ketbra{\vphi_i}{\vphi_i} \, :  1 \leq i \leq 2L^3\} 
\]
satisfies
\[ 1- \P(E_{\psi,t}) \leq \exp \left(- L \right) + 2d \exp \left(-t^2 L^2/8  \right) .\]
\end{lemma}

\noindent We apply Lemma \ref{lemma:approximation-D-by-random} with $t=\e/8d$. 
When the event $E_{\psi,t}$ holds, we have
\begin{equation} \label{equation:approx2}
(1 - \alpha) \bra{\psi} A \ket{\psi} \leq t\|A\|_{\textnormal{Tr}} + \sup_{\vphi_i \in \cN} \bra{\vphi_i} A \ket{\vphi_i}. 
\end{equation}
If the events $E_{\psi,t}$ hold simultaneously for every $\psi \in \mathcal{M}$, 
we can conclude from \eqref{equation:approx1} and \eqref{equation:approx2}
that
\begin{equation} \label{equation:approx3} 
(1-\e/2)(1 - \alpha) \lambda_1(A) \leq t \|A\|_{\textnormal{Tr}} + \sup_{\vphi_i \in \cN} \bra{\vphi_i} A \ket{\vphi_i}
\end{equation} 
Since $A$ has trace zero, we have $\|A\|_{\textnormal{Tr}} \leq 2d \lambda_1(A)$, 
and so \eqref{equation:approx3} combined with \eqref{eq:average-cap-weight}  implies that
\begin{equation*} \label{equation:approx4} 
(1-\e) \lambda_1(A) \leq \big( (1-\e/2)(1 - \alpha) -2td \big)  \lambda_1(A) 
\leq \sup_{\vphi_i \in \cN} \bra{\vphi_i} A \ket{\vphi_i},
\end{equation*}
yielding \eqref{eq:bound-lambdamax}. 
The Proposition will follow once we establish that, with positive probability,  
the events $E_{\psi,t}$ hold simultaneously for all $\psi \in \mathcal{M}$. 
To that end, we use Lemma \ref{lemma:approximation-D-by-random},  
 the estimate $\card \mathcal{M} \leq (12d/\e)^{2d}$ from 
Lemma \ref{lemma:net}, and the union bound
\begin{eqnarray} \label{eq:union-bound1} \P \left( \bigcap_{\psi \in \mathcal{M}} E_{\psi,t} \right) & \geq &
1 - \sum_{\psi \in \mathcal{M}} (1- \P(E_{\psi,t})) \\ \label{eq:union-bound2}
& \geq  & 1 - \left( \frac{12d}{\e} \right)^{2d} \left( \exp \left(- L \right) + 
2d \exp \left(-\e^2 d^{-2} L^2/512  \right)
\right).
\end{eqnarray}
We know from Lemma \ref{lemma:net}
that $\exp(c_1(\e) d) \leq L \leq \exp(C_1(\e) d)$ for some constants
$c_1(\e)$, $C_1(\e)$ depending only on $\e$. 
It follows that the quantity in \eqref{eq:union-bound1}--\eqref{eq:union-bound2}
is positive for $d$ large enough (depending on $\e$), yielding a family of $2L^3 \leq 2 \exp(3C_1(\e) d)$ 
vectors satisfying the conclusion of
Proposition \ref{proposition:approximation-D}. 
Small values of $d$ are taken into account by adjusting the constant
$C(\e)$ if necessary.
\end{proof}

\begin{proof}[Proof of Lemma \ref{lemma:approximation-D-by-random}]
Let $M_\psi = \card (\cN \cap C(\psi,\theta))$. 
The random variable $M_\psi$ follows the binomial 
distribution $B(N,p)$ for $N=2L^3$ and $p=1/L$. It follows from Hoeffding's inequality \cite{Hoeffding63}
that
\[ \P \left( B(N,p) \leq \frac{Np}{2} \right) \leq \exp \left( - \frac{p^2N}{2} \right). \]
Specialized to our situation, this yields
\begin{equation} \label{eq:bound_Mx} 
\P \left( M_\psi \leq L^2 \right) \leq \exp \left( - L \right).
\end{equation}
Moreover, conditionally on the value of $M_\psi$, the points from $\cN \cap C(\psi,\theta)$ have the same
distribution as $(\vphi_k)_{1 \leq k \leq M_\psi}$, where $(\vphi_k)$ 
are independent and uniformly distributed inside $C(\psi,\theta)$.
The random matrices 
\[
X_k=\ketbra{\vphi_k}{\vphi_k} - \E \ketbra{\vphi_1}{\vphi_1} 
= \ketbra{\vphi_k}{\vphi_k} - (1-\alpha) \bullet \ketbra{\psi}{\psi} 
\] 
(cf.\ \eqref{eq:average-cap})  are 
independent mean zero matrices. 
We now use the matrix Hoeffding inequality (see, e.g., Theorem 1.3 in \cite{Tropp12}) 
to conclude that, for any $t \geq 0$,
\begin{equation} \label{eq:matrix-hoeffding}
 \P \left( \left\| \frac{1}{M_\psi} \sum_{k=1}^{M_\psi} X_k \right\|_{\op} \geq t \right) \leq 2d \exp(-M_{\psi}t^2/8)  
\end{equation} 
(the factor $2$ appears because we want to control the operator norm rather than the 
largest eigenvalue).  
Define a random matrix $\Delta$ by the relation
\begin{equation*} \label{eq:def-Deltax}
\frac{1}{M_\psi} \sum_{k=1}^{M_\psi} \ketbra{\vphi_k}{\vphi_k}  + \Delta = (1-\alpha) \bullet \ketbra{\psi}{\psi} . 
\end{equation*}
The bound \eqref{eq:matrix-hoeffding} translates then into
$\P( \|\Delta\|_{\op} \geq t) \leq 2d \exp(-M_\psi t^2/8)$. If we remove the conditioning on $M_\psi$ and take
\eqref{eq:bound_Mx} into account, we are led to
\[ \P ( \|\Delta\|_{\op} \geq t ) \leq \exp \left( - L \right) + 2d \exp \left(-L^2 t^2/8 
\right),\]
whence  Lemma \ref{lemma:approximation-D-by-random} follows. 
\end{proof}

\section*{Conclusions}

As a consequence of Milman's tangible version of Dvoretzky's theorem, 
we gave an illustration of the complexity of entanglement in high dimensions 
by showing that the set of separable states requires
a super-exponential number of entanglement witnesses to be approximated within a constant factor, 
independent of the dimension of the instance.   
To the best of our knowledge, this is the first (unconditional) result of this nature that doesn't collapse 
in  presence of substantial randomizing noise. 
 
\smallskip 
 
There are several possible directions in which this work can be continued. 

\smallskip 

{\bfseries Upper bounds}. Are there matching upper bounds on the cardinal of minimal 
universal families of positive maps, in the sense of Theorem \ref{theorem:main}? 
One upper bound is $\exp(d_F(\cS))$, so the question is essentially equivalent to computing the
facial dimension of $\cS(\C^d \otimes \C^d)$. Since its linear dimension is $d^4-1$, 
an upper bound of $O(d^4)$ follows from general arguments.
Closing the gap between this upper bound and the lower estimate $\Omega(d^3/\log d)$ 
seems an interesting question. While related constructions using nets were considered 
by various authors, it seems likely that having a new class of invariants to focus on 
may lead to sharper results. 

\smallskip 

{\bfseries More/less robust entanglement}. 
For which values of $\e_d$ in Theorem 1' can we conclude that 
$N \gg 1$? Note that this question is meaningful only for $\e_d > \frac{2}{2+d^2}$ 
since any $\rho \in \cD(\C^d \otimes \C^d)$ has the
property that $\frac{2}{2+d^2} \bullet \rho$ is separable \cite{VidalTarrach99}. 
In the opposite direction, the reader will notice that the assertions of Theorem \ref{theorem:main} 
and  Theorem 1' differ rather modestly, and that we do not  
get further strengthening if we let $\e$ to be close to $1$ (i.e., if we insist that 
the family of witnesses be universal for all $\rho$ such that  
$(1-\delta)\bullet \rho$ is entangled for some small $\delta >0$, 
or even if we let $\delta=\delta_d \to 0$ when $d \to \iy$). 
This is because our lower bound on $d_F(\cS,A)$ does not improve substantially when $A\to 1$. 
However, it is still conceivable that, for some other 
general reasons, $\cS$ is harder to ``finely-approximate'' by a polytope with few faces 
than $\cD$, which would perhaps allow to retrieve the results of \cite{Skowronek16} 
from general principles, and to link the perspective provided by our approach with the 
prior algorithmic 
 results on entanglement detection.  

\smallskip 
 
 {\bfseries Multipartite or unbalanced setting}.   What if the underlying Hilbert space is 
of the form $\cH=\C^d \otimes \C^m$  or $\cH=\big(\C^d\big)^{\otimes N}$? 

\smallskip 

{\bfseries Unbalanced witnesses}. What if we use witnesses $\Phi : \M_d \to \M_m$, where $m=\mathrm{poly}(d)$?

\medskip

Finally, our primary motivation was to bring to the attention of 
the quantum information and theoretical computer science
 communities another tool from asymptotic geometric analysis, which didn't seem to be 
 widely known. Given the strong algorithmic flavor of the (40 years old!) 
 Figiel--Lindenstrauss--Milman inequality \eqref{eq:FLM}, 
it is quite likely that it has applications to complexity theory that go beyond entanglement detection. 

\section*{Acknowledgements}

The authors would like to thank the Instituto de Ciencias Matem\'aticas in Madrid where this work was initiated.

\bibliographystyle{amsplain}

\begin{thebibliography}{99}

\bibitem{AubrunLancien15}
Guillaume Aubrun and C{\'e}cilia Lancien.
\newblock Locally restricted measurements on a multipartite quantum system:
  data hiding is generic.
\newblock {\em Quantum Inf. Comput.}, 15(5-6):513--540, 2015.

\bibitem{book}
Guillaume Aubrun and Stanis{\l}aw Szarek.
\newblock {\em Alice and Bob Meet Banach. The Interface of Asymptotic Geometric
  Analysis and Quantum Information Theory}.
\newblock To appear in {\em Mathematical Surveys and Monographs}, 
Amer. Math. Soc., Providence.   

\bibitem{AubrunSzarek06}
Guillaume Aubrun and Stanis{\l}aw~J. Szarek.
\newblock Tensor products of convex sets and the volume of separable states on
  $n$ qudits.
\newblock {\em Phys. Rev. A}, 73(2):022109, 2006.

\bibitem{ASW11}
Guillaume Aubrun, Stanis{\l}aw Szarek, and Elisabeth Werner.
\newblock Hastings's additivity counterexample via {D}voretzky's theorem.
\newblock {\em Comm. Math. Phys.}, 305(1):85--97, 2011.

\bibitem{ASY12}
Guillaume Aubrun, Stanis\l{}aw~J. Szarek, and Deping Ye.
\newblock Phase transitions for random states and a semicircle law for the
  partial transpose.
\newblock {\em Phys. Rev. A}, 85(3):030302(R), 2012.

\bibitem{ASY14}
Guillaume Aubrun, Stanis{\l}aw~J. Szarek, and Deping Ye.
\newblock Entanglement thresholds for random induced states.
\newblock {\em Comm. Pure Appl. Math.}, 67(1):129--171, 2014.

\bibitem{Barvinok14}
Alexander {Barvinok}.
\newblock {Thrifty approximations of convex bodies by polytopes.}
\newblock {\em {Int. Math. Res. Not.}}, 2014(16):4341--4356, 2014.

\bibitem{Bennett93}
Charles~H. Bennett, Gilles Brassard, Claude Cr\'epeau, Richard Jozsa, Asher
  Peres, and William~K. Wootters.
\newblock Teleporting an unknown quantum state via dual classical and
  {E}instein--{P}odolsky--{R}osen channels.
\newblock {\em Phys. Rev. Lett.}, 70:1895--1899, Mar 1993.

\bibitem{BCY11}
Fernando~G.S.L. {Brand\~ao}, Matthias {Christandl}, and Jon {Yard}.
\newblock {A quasipolynomial-time algorithm for the quantum separability
  problem.}
\newblock In {\em {Proceedings of the 43rd annual ACM symposium on theory of
  computing, STOC '11. San Jose, CA, USA, June 6--8, 2011.}}, pages 343--352.
  New York, NY: Association for Computing Machinery (ACM), 2011.

\bibitem{Bronstein08}
E.M. Bronstein.
\newblock Approximation of convex sets by polytopes. 
(Russian) 
\newblock {\em Sovrem. Mat. Fundam. Napravl.} 22:5--37, 2007. 
\newblock Translation in 
{\em J. Math. Sci. (N. Y.)} 153(6):727--762, 2008.

\bibitem{DPS04}
Andrew~C. Doherty, Pablo~A. Parrilo, and Federico~M. Spedalieri.
\newblock Complete family of separability criteria.
\newblock {\em Phys. Rev. A}, 69:022308, Feb 2004.

\bibitem{Dvoretzky61}
Aryeh Dvoretzky.
\newblock Some results on convex bodies and {B}anach spaces.
\newblock In {\em Proc. {I}nternat. {S}ympos. {L}inear {S}paces ({J}erusalem,
  1960)}, pages 123--160. Jerusalem Academic Press, Jerusalem; Pergamon,
  Oxford, 1961.

\bibitem{EPR35}
A.~Einstein, B.~Podolsky, and N.~Rosen.
\newblock Can quantum-mechanical description of physical reality be considered
  complete?
\newblock {\em Phys. Rev.}, 47:777--780, May 1935.

\bibitem{FLM77}
T.~Figiel, J.~Lindenstrauss, and V.~D. Milman.
\newblock The dimension of almost spherical sections of convex bodies.
\newblock {\em Acta Math.}, 139(1-2):53--94, 1977.

\bibitem{Gharibian10}
Sevag Gharibian.
\newblock Strong {NP}-hardness of the quantum separability problem.
\newblock {\em Quantum Inf. Comput.}, 10(3-4):343--360, 2010.

\bibitem{Gordon88}
Y.~Gordon.
\newblock On {M}ilman's inequality and random subspaces which escape through a
  mesh in {${\bf R}^n$}.
\newblock In {\em Geometric aspects of functional analysis (1986/87)}, volume
  1317 of {\em Lecture Notes in Math.}, pages 84--106. Springer, Berlin, 1988.

\bibitem{Gurvits03}
Leonid Gurvits.
\newblock Classical deterministic complexity of {E}dmonds' problem and quantum
  entanglement.
\newblock In {\em Proceedings of the thirty-fifth annual ACM symposium on
  Theory of computing}, pages 10--19. ACM, 2003.

\bibitem{GurvitsBarnum02}
Leonid Gurvits and Howard Barnum.
\newblock Largest separable balls around the maximally mixed bipartite quantum
  state.
\newblock {\em Phys. Rev. A}, 66(6):062311, 2002.

\bibitem{GHMW15}
Gus Gutoski, Patrick Hayden, Kevin Milner, and Mark~M. Wilde.
\newblock Quantum interactive proofs and the complexity of separability
  testing.
\newblock {\em Theory Comput.}, 11:59--103, 2015.

\bibitem{HaKye11}
Kil-Chan {Ha} and Seung-Hyeok {Kye}.
\newblock {Entanglement witnesses arising from exposed positive linear maps.}
\newblock {\em {Open Syst. Inf. Dyn.}}, 18(4):323--337, 2011.

\bibitem{HarrowMontanaro13}
Aram~W. Harrow and Ashley Montanaro.
\newblock Testing product states, quantum {M}erlin-{A}rthur games and tensor
  optimization.
\newblock {\em J. ACM}, 60(1):3:1--3:43, February 2013.

\bibitem{HNW16}
Aram~W. Harrow, Anand Natarajan, and Xiaodi Wu.
\newblock Limitations of semidefinite programs for separable states and entangled games, 
\newblock   	arXiv:1612.09306 [quant-ph]. 

\bibitem{Hastings09}
Matthew~B Hastings.
\newblock Superadditivity of communication capacity using entangled inputs.
\newblock {\em Nature Physics}, 5(4):255--257, 2009.

\bibitem{Hoeffding63}
W.~{Hoeffding}.
\newblock {Probability inequalities for sums of bounded random variables.}
\newblock {\em {J. Amer. Statist. Assoc.}}, 58:13--30, 1963.

\bibitem{HHH96}
{Micha\l} Horodecki, {Pawe\l} Horodecki, and Ryszard Horodecki.
\newblock Separability of mixed states: necessary and sufficient conditions.
\newblock {\em Phys. Lett. A}, 223(1--2):1--8, 1996.

\bibitem{PH97}
Pawel Horodecki.
\newblock Separability criterion and inseparable mixed states with positive
  partial transposition.
\newblock {\em Phys. Lett. A}, 232(5):333 -- 339, 1997.

\bibitem{Ioannou07}
Lawrence~M. Ioannou.
\newblock Computational complexity of the quantum separability problem.
\newblock {\em Quantum Inf. Comput.}, 7(4):335--370, 2007.

\bibitem{Lancien16}
C\'ecilia Lancien.
\newblock $k$-extendibility of high-dimensional bipartite quantum states.
\newblock {\em Random Matrices Theory Appl.}, 05(03):1650011, 2016.

\bibitem{Levy51}
Paul L{\'e}vy.
\newblock {\em Probl\`emes concrets d'analyse fonctionnelle. {A}vec un
  compl\'ement sur les fonctionnelles analytiques par {F}. {P}ellegrino}.
\newblock Gauthier-Villars, Paris, 1951.
\newblock 2d ed.

\bibitem{LRT14}
Alexander~E. Litvak, Mark Rudelson, and Nicole Tomczak-Jaegermann.
\newblock On approximation by projections of polytopes with few facets.
\newblock {\em Israel J. Math.}, pages 1--20, 2014.

\bibitem{Milman71}
V.~D. Milman.
\newblock A new proof of {A}. {D}voretzky's theorem on cross-sections of convex
  bodies.
\newblock {\em Funkcional. Anal. i Prilo\v zen.}, 5(4):28--37, 1971.

\bibitem{Peres96}
Asher Peres.
\newblock Separability criterion for density matrices.
\newblock {\em Phys. Rev. Lett.}, 77:1413--1415, Aug 1996.

\bibitem{Pisier89}
Gilles Pisier.
\newblock {\em The volume of convex bodies and {B}anach space geometry},
  volume~94 of {\em Cambridge Tracts in Mathematics}.
\newblock Cambridge University Press, Cambridge, 1989.

\bibitem{Rogers57}
C.~A. Rogers.
\newblock A note on coverings.
\newblock {\em Mathematika}, 4:1--6, 1957.

\bibitem{Rogers63}
C.~A. Rogers.
\newblock Covering a sphere with spheres.
\newblock {\em Mathematika}, 10:157--164, 1963.

\bibitem{Schechtman89}
Gideon Schechtman.
\newblock A remark concerning the dependence on {$\epsilon$} in {D}voretzky's
  theorem.
\newblock In {\em Geometric aspects of functional analysis (1987--88)}, volume
  1376 of {\em Lecture Notes in Math.}, pages 274--277. Springer, Berlin, 1989.

\bibitem{Schrodinger35}
E.~Schr\"{o}dinger.
\newblock Discussion of probability relations between separated systems.
\newblock {\em Math. Proc. Cambridge Philos. Soc.}, 31:555--563, 10 1935.

\bibitem{Skowronek16}
{\L}ukasz Skowronek.
\newblock There is no direct generalization of positive partial transpose
  criterion to the three-by-three case.
\newblock {\em arXiv preprint 1605.05254}, 2016.

\bibitem{Stormer63}
Erling St{\o}rmer.
\newblock Positive linear maps of operator algebras.
\newblock {\em Acta Math.}, 110:233--278, 1963.

\bibitem{Szarek14}
Stanis{\l}aw~J. Szarek.
\newblock Coarse approximation of convex bodies by polytopes and
  the complexity of {B}anach--{M}azur compacta. 
\newblock Preprint 2014.

\bibitem{SWZ08}
Stanis{\l}aw~J. Szarek, Elisabeth Werner, and Karol {\.Z}yczkowski.
\newblock Geometry of sets of quantum maps: a generic positive map acting on a
  high-dimensional system is not completely positive.
\newblock {\em J. Math. Phys.}, 49(3):032113, 21, 2008.

\bibitem{Tropp12}
Joel~A Tropp.
\newblock User-friendly tail bounds for sums of random matrices.
\newblock {\em Found. Comput. Math.}, 12(4):389--434, 2012.

\bibitem{VidalTarrach99}
Guifre Vidal and Rolf Tarrach.
\newblock Robustness of entanglement.
\newblock {\em Phys. Rev. A}, 59(1):141, 1999.

\bibitem{Werner89}
Reinhard~F. Werner.
\newblock Quantum states with {E}instein--{P}odolsky--{R}osen correlations
  admitting a hidden-variable model.
\newblock {\em Phys. Rev. A}, 40:4277--4281, Oct 1989.

\end{thebibliography}

\begin{dajauthors}
\begin{authorinfo}[ga]
  Guillaume Aubrun\\
  Universit\'e Lyon 1\\
  Lyon, France\\
  aubrun\imageat{}math\imagedot{}univ-lyon1\imagedot{}fr \\
  \url{http://math.univ-lyon1.fr/~aubrun}
\end{authorinfo}
\begin{authorinfo}[sjs]
  Stanis\l aw Szarek\\
  Case Western Reserve University \& Universit\'{e}  Pierre et Marie Curie\\
  Cleveland, USA \& Paris, France\\
  szarek\imageat{}cwru\imagedot{}edu \\
  \url{http://www.cwru.edu/artsci/math/szarek}
\end{authorinfo}
\end{dajauthors}

\end{document}